\documentclass[11pt]{llncs}
\usepackage{graphicx,amsmath,amssymb,enumerate}
\usepackage{algorithm}
\usepackage{algorithmic}
\usepackage{subfiles}
\usepackage{titling}
\usepackage{tikz}
\usetikzlibrary{shapes.geometric, arrows}

\usepackage{fullpage}
\usepackage{authblk}

\usepackage{amsmath}
\usepackage{xcolor}
\usepackage{graphicx}
\usepackage[export]{adjustbox}

\newcommand{\Patrascu }{P\v{a}tra\c{s}cu}
\newcommand{\modulo}{\operatorname{mod}{}}

\newcommand{\remove}[1]{}

\newcommand{\ThreeSUM}{\textsf{3SUM}}
\newcommand{\ThreeSUMIndexing}{\textsf{3SUM-Indexing}}
\newcommand{\KplusOneSUMIndexing}{\textsf{(k+1)-SUM-Indexing}}
\newcommand{\TSI}{\textsf{3SI}}

\newcommand{\SetDisjointness}{\textsf{SetDisjointness}}
\newcommand{\KSetDisjointness}{\textsf{k-SetDisjointness}}

\newcommand{\SetIntersection}{\textsf{SetIntersection}}

\newcommand{\KReachability}{\textsf{k-Reachability}}
\newcommand{\Reachability}{\textsf{Reachability}}
\newcommand{\KmOneReachability}{\textsf{(k-1)-Reachability}}
\newcommand*\samethanks[1][\value{footnote}]{\footnotemark[#1]}

\begin{document}

\title{Conditional Lower Bounds for Space/Time Tradeoffs}
\date{}

\author[1]{Isaac Goldstein\thanks{This research is supported by the Adams Foundation of the Israel Academy of Sciences and Humanities}}
\author[2]{Tsvi Kopelowitz\thanks{Part of this work took place while the second author was at University of Michigan.  This work is supported in part by the Canada Research Chair for Algorithm Design,  NSF grants CCF-1217338,  CNS-1318294,  and CCF-1514383}}
\author[1]{Moshe Lewenstein\thanks{This work was partially supported by an ISF grant \#1278/16}}
\author[1]{Ely Porat \samethanks}
\affil[1]{Bar-Ilan University \\  \texttt{\{goldshi,moshe,porately\}@cs.biu.ac.il}}
\affil[2]{University of Waterloo \\  \texttt{kopelot@gmail.com}}
\maketitle

\thispagestyle{empty}

\begin{abstract}
In recent years much effort has been concentrated towards achieving polynomial time lower bounds on algorithms for solving various well-known problems. A useful technique for showing such lower bounds is to prove them conditionally based on well-studied hardness assumptions such as 3SUM, APSP, SETH, etc. This line of research helps to obtain a better understanding of the complexity inside P.

A related question asks to prove conditional \emph{space} lower bounds on data structures that are constructed to solve certain algorithmic tasks after an initial preprocessing stage. This question received little attention in previous research even though it has potential strong impact.

In this paper we address this question and show that surprisingly many of the well-studied hard problems that are known to have conditional polynomial \emph{time} lower bounds are also hard when concerning \emph{space}. This hardness is shown as a tradeoff between the space consumed by the data structure and the time needed to answer queries. The tradeoff may be either smooth or admit one or more singularity points.

We reveal interesting connections between different space hardness conjectures and present matching upper bounds. We also apply these hardness conjectures to both static and dynamic problems and prove their conditional space hardness.

We believe that this novel framework of polynomial space conjectures can play an important role in expressing polynomial space lower bounds of many important algorithmic problems. Moreover, it seems that it can also help in achieving a better understanding of the hardness of their corresponding problems in terms of time.
\end{abstract}

\section{Introduction}

\subsection{Background}
Lately there has been a concentrated effort to understand the time complexity within P, the class of decision problems solvable by polynomial time algorithms. The main goal is to explain why certain problems have time complexity that seems to be non-optimal. For example, all known efficient algorithmic solutions for the 3SUM problem, where we seek to determine whether there are three elements $x,y,z$ in input set $S$ of size $n$ such that $x+y+z=0$, take $\tilde{O}(n^2)$ time\footnote{The $\tilde{O}$ and $\tilde{\Omega}$ notations suppress polylogarithmic factors}. However, the only real lower bound that we know is the trivial $\Omega(n)$. Likewise, we know how to solve the {\em all pairs shortest path}, APSP, problem in $\tilde{O}(n^3)$ time but we cannot even determine whether it is impossible to obtain an $\tilde{O}(n^2)$ time algorithm.
One may note that it follows from the time-hierarchy theorem that there exist problems in P with complexity $\Omega(n^k)$ for every fixed $k$. Nevertheless, such a separation for natural practical problems seems to be hard to achieve.

The collaborated effort to understand the internals of P has been concentrated on identifying some basic problems that are conjectured to be hard to solve  more efficiently (by polynomial factors) than their current known complexity.
These problems serve as a basis to prove conditional hardness of other problems by using reductions. The reductions are reminiscent of NP-complete reductions but differ in that they are restricted to be of time complexity strictly smaller (by a polynomial factor) than the problem that we are reducing to.
Examples of such hard problems include the well-known 3SUM problem, the fundamental APSP problem, (combinatorial) Boolean matrix multiplication, etc. Recently, conditional time lower bounds have been proven based on the conjectured hardness of these problems for graph algorithms~\cite{AGW15,WW10}, edit distance~\cite{BI15}, longest common subsequence (LCS)~\cite{ABW15,BK15}, dynamic algorithms~\cite{AW14,Patrascu10}, jumbled indexing~\cite{ACLL14}, and many other problems~\cite{ABHWZ16,ABW15a,AWW14,AWY15,Bringmann14,HKNS15,KPP16,LMNT15,PW10}. 

\subsection{Motivation}
In stark contrast to polynomial \emph{time} lower bounds, little effort has been devoted to finding polynomial \emph{space} conditional lower bounds. %\tnote{This is a dangerous statement, and likely not even true. We need to phrase this very differently.}
An example of a space lower bound appears in the work of Cohen and Porat~\cite{CP10a} and P\v{a}tra\c{s}cu and Roditty~\cite{PR14} where lower bounds are shown on the size of a distance oracle for sparse graphs based on a conjecture about the best possible data structure for a set intersection problem (which we call set disjointness in order to differ it from its reporting variant).

A more general question is, for algorithmic problems, what conditional lower bounds of a space/time tradeoff can be shown based on the set disjointness (intersection) conjecture? Even more general is to discover what space/time tradeoffs can be achieved based on the other algorithmic problems that we assumed are hard (in the time sense)? Also, what are the relations between these identified "hard" problems in the space/time tradeoff sense? These are the questions which form the basis and framework of this paper.

Throughout this paper we show connections between different hardness assumptions, show some matching upper bounds and propose several conjectures based on this accumulated knowledge. Moreover, we conjecture that there is a strong correlation between polynomial hardness in time and space. We note that in order to discuss space it is often more natural to consider data structure variants of problems and this is the approach we follow in this paper.

\subsection{Our Results}

\paragraph{Set Disjointness.}
In the \SetDisjointness{} problem mentioned before, it is required to preprocess a collection of $m$ sets $S_1,\cdots, S_m \subset U$, where $U$ is the universe of elements and the total number of elements in all sets is $N$. For a query, a pair of integers $(i,j)$ $(1\leq i,j \leq m)$ is given and we are asked whether $S_i \cap S_j$ is empty or not. A folklore conjecture, which appears in~\cite{CP10,PR14}, suggests that to achieve a constant query time the space of the data structure constructed in the preprocessing stage needs to be $\tilde{\Omega}(N^{2})$. We call this conjecture the \SetDisjointness{} conjecture. This conjecture does not say anything about the case where we allow \emph{higher} query time. Therefore, we suggest a stronger conjecture which admits a \emph{full tradeoff} between the space consumed by the data structure (denoted by $S$) and the query time (denoted by $T$). This is what we call the Strong \SetDisjointness{} conjecture. This conjecture states that for solving \SetDisjointness{} with a query time $T$ our data structure needs $\tilde{\Omega}(N^{2}/T^2)$ space. A matching upper bound exists for this problem by generalizing ideas from~\cite{CP10} (see also~\cite{KPP15}). Our new \SetDisjointness{} conjecture can be used to admit more expressive space lower bounds for a full tradeoff between space and query time.

\paragraph{3SUM Indexing.}
One of the basic and frequently used hardness conjectures is the celebrated 3SUM conjecture. This conjecture was used for about 20 years to show many conditional \emph{time} lower bounds on various problems. However, we focus on what can be said about its \emph{space} behavior. To do this, it is natural to consider a data structure version of 3SUM which allows one to preprocess the input set $S$. Then, the query is an external number $z$ for which we need to answer whether there are $x,y\in S$ such that $x+y=z$. It was pointed out by Chan and Lewenstein~\cite{CL15} that all known algorithms for 3SUM actually work within this model as well. We call this problem {\em 3SUM Indexing}. On one hand, this problem can easily be solved using $O(n^2)$ space by sorting $x+y$ for all $x,y\in S$ and then searching for $z$ in $\tilde{O}(1)$ time. On the other hand, by just sorting $S$ we can answer queries by a well-known linear time algorithm. The big question is whether we can obtain better than $\tilde{\Omega}(n^2)$ space while using just $\tilde{O}(1)$ time query? Can it be done even if we allow $\tilde{O}(n^{1-\Omega(1)})$ query time? This leads us to our two new hardness conjectures. The \ThreeSUMIndexing{} conjecture states that when using $\tilde{O}(1)$ query time we need $\tilde{\Omega}(n^2)$ space to solve \ThreeSUMIndexing{}. In the Strong \ThreeSUMIndexing{} conjecture we say that even when using $\tilde{O}(n^{1-\Omega(1)})$ query time we need $\tilde{\Omega}(n^2)$ space to solve \ThreeSUMIndexing{}.

\paragraph{3SUM Indexing and Set Disjointness.}
We prove connections between the \SetDisjointness{} conjectures and the \ThreeSUMIndexing{} conjectures. Specifically, we show that the Strong \ThreeSUMIndexing{} conjecture implies the Strong \SetDisjointness{} conjecture, while the \SetDisjointness{} conjecture implies the \ThreeSUMIndexing{} conjecture. This gives some evidence towards establishing the difficulty within the \ThreeSUMIndexing{} conjectures. The usefulness of these conjectures should not be underestimated. As many problems are known to be 3SUM-hard these new conjectures can play an important role in achieving \emph{space} lower bounds on their corresponding data structure variants. Moreover, it is interesting to point on the difference between \SetDisjointness{} which admits smooth tradeoff between space and query time and \ThreeSUMIndexing{} which admits a big gap between the two trivial extremes. This may explain why we are unable to show full equivalence between the hardness conjectures of the two problems. Moreover, it can suggest a separation between problems with smooth space-time behavior and others which have no such tradeoff but rather two "far" extremes.

\paragraph{Generalizations.}
Following the discussion on the \SetDisjointness{} and the \ThreeSUMIndexing{} conjectures we investigate their generalizations.

\paragraph{I. k-Set Disjointness and (k+1)-SUM Indexing.}
The first generalization is a natural parametrization of both problems. In the \SetDisjointness{} problem we query about the emptiness of the intersection between \emph{two} sets, while in the \ThreeSUMIndexing{} problem we ask, given a query number $z$, whether \emph{two} numbers of the input $S$ sum up to $z$. In the parameterized versions of these problems we are interested in the emptiness of the intersection between \emph{k} sets and ask if \emph{k} numbers sum up to a number given as a query. These generalized variants are called \KSetDisjointness{} and \KplusOneSUMIndexing{} respectively. For each problem we give corresponding space lower bounds conjectures which generalize those of \SetDisjointness{} and \ThreeSUMIndexing{}. These conjectures also have corresponding strong variants which are accompanied by matching upper bounds. We prove that the \KSetDisjointness{} conjecture implies \KplusOneSUMIndexing{} conjecture via a novel method using linear equations.

\paragraph{II. k-Reachability.}
A second generalization is the problem we call \KReachability{}. In this problem we are given as an input a directed sparse graph $G=(V,E)$ for preprocessing. Afterwards, for a query, given as a pair of vertices $u,v$, we wish to return if there is a path from $u$ to $v$ consisting of at most $k$ edges. We provide an upper bound on this problem for every fixed $k \geq 1$. The upper bound admits a tradeoff between the space of the data structure (denoted by $S$) and the query time (denoted by $T$), which is $ST^{2/(k-1)}=O(n^2)$. We argue that this upper bound is tight. That is, we conjecture that if query takes $T$ time, the space must be $\tilde{\Omega}(\frac{n^2}{T^{2/(k-1)}})$. We call this conjecture the \KReachability{} conjecture.

We give three indications towards the correctness of this conjecture.
First, we prove that the base case, where $k=2$, is equivalent to the \SetDisjointness{} problem. This is why this problem can be thought of as a generalization of \SetDisjointness{}.

Second, if we consider non-constant $k$ then the smooth tradeoff surprisingly disappears and we get "extreme behavior" as $\tilde{\Omega}(\frac{n^2}{T^{2/(k-1)}})$ eventually becomes $\tilde{\Omega}(n^2)$. This means that to answer reachability queries for non-constant path length, we can either store all answers in advance using $n^2$ space or simply answer queries from scratch using a standard graph traversal algorithm. The general problem where the length of the path from $u$ to $v$ is unlimited in length is sometimes referred to as the problem of constructing efficient reachability oracles. \Patrascu{} in~\cite{Patrascu11} leaves it as an open question if a data structure with less than $\tilde{\Omega}(n^2)$ space can answer reachability queries efficiently. Moreover, \Patrascu{} proved that for constant time query, truly superlinear space is needed. Our  \KReachability{} conjecture points to this direction, while admitting full space-time tradeoff for constant $k$.

The third indication for the correctness of the \KReachability{} conjecture comes from a connection to distance oracles. A {\em distance oracle} is a data structure that can be used to quickly answer queries about the shortest path between two given nodes in a preprocessed undirected graph. As mentioned above, the \SetDisjointness{} conjecture was used to exclude some possible tradeoffs for sparse graphs. Specifically, Cohen and Porat~\cite{CP10a} showed that obtaining an approximation ratio smaller than 2 with constant query time requires $\tilde{\Omega}(n^2)$ space. Using a somewhat stronger conjecture \Patrascu{} and Roditty~\cite{PR14} showed that a (2,1)-distance oracle for unweighted graphs with $m = O(n)$ edges requires $\tilde{\Omega}(n^{1.5})$ space. Later, this result was strengthened by P\v{a}tra\c{s}cu et al.~\cite{PRT12}. However, these results do not exclude the possibility of compact distance oracles if we allow higher query time. For stretch-2 and stretch-3 in sparse graphs, Agarwal et. al.~\cite{AGH12,AGH11} achieved a space-time tradeoff of $S \times T = O(n^2)$ and $S \times T^2 = O(n^2)$, respectively. Agarwal~\cite{Agarwal14} also showed many other results for stretch-2 and below. We use our \KReachability{} conjecture to prove that for stretch-less-than-(1+2/k) distance oracles $S \times T^{2/(k-1)}$ is bounded by $\tilde{\Omega}(n^2)$. This result is interesting in light of Agarwal~\cite{Agarwal14} where a stretch-(5/3) oracle was presented which achieves a space-time tradeoff of $S \times T = O(n^2)$. This matches our lower bound, where $k=3$, if our lower bound would hold not only for stretch-less-than-(5/3) but also for stretch-(5/3) oracles.
Consequently, we see that there is strong evidence for the correctness of the \KReachability{} conjecture.

Moreover, these observations show that on one hand \KReachability{} is a generalization of \SetDisjointness{} which is closely related to \ThreeSUMIndexing{}. On the other hand, \KReachability{} is related to distance oracles which solve the famous APSP problem using smaller space by sacrificing the accuracy of the distance between the vertices. Therefore, the \KReachability{} conjecture seems as a conjecture corresponding to the APSP hardness conjecture, while also admitting some connection with the celebrated 3SUM hardness conjecture.

\paragraph{SETH and Orthogonal Vectors.}
After considering space variants of the 3SUM and APSP conjectures it is natural to consider space variants for the Strong Exponential Time Hypothesis (SETH) and the closely related conjecture of orthogonal vectors. SETH asserts that for any $\epsilon > 0$ there is an integer $k > 3$ such that k-SAT cannot be solved in $2^{(1-\epsilon)n}$ time. The orthogonal vectors time conjecture states that there is no algorithm that for every $c \geq 1$, finds if there are at least two orthogonal vectors in a set of $n$ Boolean vectors of length $c\log{n}$ in $\tilde{O}(n^{2-\Omega(1)})$ time. We discuss the space variants of these conjectures in Section~\ref{sec:SETH_OV}. However, we are unable to connect these conjectures and the previous ones. This is perhaps not surprising as the connection between SETH and the other conjectures even in the time perspective is very loose (see, for example, discussions in ~\cite{AW14,HKNS15}).

\paragraph{Boolean Matrix Multiplication.}
Another problem which receives a lot of attention in the context of conditional time lower bounds is calculating Boolean Matrix Multiplication (BMM). We give a data structure variant of this well-known problem. We then demonstrate the connection between this problem and the problems of \SetDisjointness{} and \KReachability{}.

\paragraph{Applications.}
Finally, armed with the \emph{space} variants of many well-known conditional \emph{time} lower bounds, we apply this conditional space lower bounds to some static and dynamic problems. This gives interesting space lower bound results on these important problems which sometimes also admits clear space-time tradeoff. We believe that this is just a glimpse of space lower bounds that can be achieved based on our new framework and that many other interesting results are expected to follow this promising route.

\bigskip
Figure~\ref{fig:M1} in Appendix~\ref{sec:sketch} presents a sketch of the results in this paper.

\section{Set Intersection Hardness Conjectures}
We first give formal definitions of the \SetDisjointness{} problem and its enumeration variant:

\begin{problem}[\SetDisjointness{} Problem]\label{prob:set_dis}
Preprocess a family $F$ of $m$ sets, all from universe $U$, with total size $N=\sum_{S\in F} |S|$ so that given two query sets $S,S'\in F$ one can
determine if $S\cap S'=\emptyset$.
\end{problem}

\begin{problem}[\SetIntersection{} Problem]\label{prob:set_int}
Preprocess a family $F$ of $m$ sets, all from universe $U$, with total size $N=\sum_{S\in F} |S|$ so that given two query sets $S,S'\in F$ one can enumerate
the set $S\cap S'$.
\end{problem}

\paragraph{Conjectures.} The \SetDisjointness{} problem was regarded as a problem that admits space hardness. The hardness conjecture of the \SetDisjointness{} problem has received several closely related formulations. One such formulation, given by \Patrascu{} and Roditty~\cite{PR14}, is as follows:

\begin{conjecture}
\textbf{\SetDisjointness{} Conjecture [Formulation 1]}. Any data structure for the \SetDisjointness{} problem where $|U| =\log^cm$ for a large enough constant $c$ and with a constant query time must use $\tilde{\Omega}(m^2)$ space.
\end{conjecture}

Another formulation is implicitly suggested in Cohen and Porat~\cite{CP10}:
\begin{conjecture}
\textbf{\SetDisjointness{} Conjecture [Formulation 2]}. Any data structure for the \SetDisjointness{} problem with constant query time must use $\tilde{\Omega}(N^2)$ space.
\end{conjecture}

There is an important distinction between the two formulations, which is related to the sparsity of \SetDisjointness{} instances.
This distinction follows from the following upper bound: store an $m\times m$ matrix of the answers to all possible queries, and then queries will cost constant time. The first formulation of the \SetDisjointness{} conjecture states that if we want constant (or poly-logaritmic) query time, then this is the best we can do. At a first glance this makes the second formulation, whose bounds are in terms of $N$ and not $m$, look rather weak. In particular, why would we ever be interested in a data structure that uses $O(N^2)$ space when we can use one with $O(m^2)$ space? The answer is that the two conjectures are the same if the sets are very sparse, and so at least in terms of $N$, if one were to require a constant query time then by the second formulation the space must be at least $\Omega(N^2)$ (which happens in the very sparse case).

Nevertheless, we present a more general conjecture, which in particular captures a tradeoff curve between the space usage and query time.
This formulation captures the difficulty that is commonly believed to arise from the \SetDisjointness{} problem, and matches the upper bounds of Cohen and Porat~\cite{CP10} (see also~\cite{KPP15}).

\begin{conjecture}\label{conj:SSDH}
\textbf{Strong \SetDisjointness{} Conjecture}. Any data structure for the \SetDisjointness{} problem that answers queries in $T$ time must use $S=\tilde{\Omega}(\frac{N^2}{T^2})$ space.
\end{conjecture}

For example, a natural question to ask is ``what is the smallest query time possible with linear space?". This question is addressed, at least from a lower bound perspective, by the Strong \SetDisjointness{} conjecture. 

\begin{conjecture}\label{conj:SSIH}
\textbf{Strong \SetIntersection{} Conjecture}. Any data structure for the \SetIntersection{} problem that answers queries in $O(T + op)$ time, where $op$ is the size of the output of the query, must use $S=\tilde{\Omega}(\frac{N^2}{T})$ space.
\end{conjecture}

\section{\ThreeSUMIndexing{} Hardness Conjectures}

In the classic 3SUM problem we are given an integer array $A$ of size $n$ and we wish to decide whether there are 3 distinct integers in $A$ which sum up to zero. Gajentaan and Overmars~\cite{GajentaanO95} showed that an equivalent formulation of this problem receives 3 integer arrays $A_1$, $A_2$, and $A_3$, each of size $n$, and the goal is to decide if there is a triplet $x_1\in A_1, x_2\in A_2$, and $x_3\in A_3$ that sum up to zero.

We consider the data structure variant of this problem which is formally defined as follows:
\begin{problem}[\ThreeSUMIndexing{} Problem]\label{prob:3sum_ind}
Preprocess two integer arrays $A_1$ and $A_2$, each of length $n$, so that given a query integer $z$ we can decide whether there are $x\in A_1$ and $y\in A_2$ such that $z=x+y$.
\end{problem}

It is straightforward to maintain all possible $O(n^2)$ sums of pairs in quadratic space, and then answer a query in $\tilde{O}(1)$ time. On the other extreme, if one does not wish to utilize more than linear space then one can sort the arrays separately during preprocssing time, and then a query can be answered in $\tilde{O}(n)$ time by scanning both of the sorted arrays in parallel and  in opposite directions.

We introduce two conjectures with regards to the \ThreeSUMIndexing{} problem, which serve as natural candidates for proving polynomial space lower bounds.

\begin{conjecture}
\textbf{\ThreeSUMIndexing{} Conjecture}: There is no solution for the \ThreeSUMIndexing{} problem with truly subquadratic space and $\tilde{O}(1)$ query time.
\end{conjecture}

\begin{conjecture}
\textbf{Strong \ThreeSUMIndexing{} Conjecture}: There is no solution for the \ThreeSUMIndexing{} problem with truly subquadratic space and truly sublinear query time.
\end{conjecture}

Notice that one can solve the classic \ThreeSUM{} problem using a data structure for \ThreeSUMIndexing{} by preprocessing $A_1$ and $A_2$, and answering $n$ \ThreeSUMIndexing{} queries on all of the values in $A_3$.

Next, we prove theorems that show tight connections between the \ThreeSUMIndexing{} conjectures and the \SetDisjointness{} conjectures.
We note that the proofs of the first two theorems are similar to the proofs of~\cite{KPP16}, but with space interpretation.

\begin{theorem}\label{thm:3SI_to_SSD}
The Strong \ThreeSUMIndexing{} Conjecture implies the Strong \SetDisjointness{} Conjecture.
\end{theorem}
\begin{proof}
A family $\mathcal{H}$ of hash functions from $[u] \rightarrow [m]$
is called {\em linear} if for any $h\in\mathcal{H}$ and any $x,x' \in [u]$,
$h(x) + h(x')= h(x+x') +c_h\; (\modulo m)$, where $c_h$ is some integer that depends only on $h$.
$\mathcal{H}$ is called {\em almost linear} if for any $h\in\mathcal{H}$ and any $x,x' \in [u]$, either
$h(x) + h(x')= h(x+x') +c_h\; (\modulo m)$, or $h(x) + h(x') = h(x+x')+c_h+1 \; (\modulo m)$.

Given a hash function $h\in \mathcal{H}$ we say that a value $i\in m$ is heavy for set $S=\{x_1,\ldots, x_n\} \subset [u]$ if  $|\{x\in S : h(x) = i\}| > \frac{3n}m$ .
$\mathcal{H}$ is called {\em almost balanced} if for any set $S=\{x_1,\ldots, x_n\} \subset [u]$, the expected number of elements from $S$ that are hashed to heavy values is $O(m)$.
Kopelowitz et al. showed in~\cite{KPP16} that a family of hash functions obtained from the construction of Dietzfelbinger~\cite{Dietzfelbinger96} is almost-linear, almost-balanced, and pair-wise independent. In order to reduce clutter in the proof here we assume the existence of linear, almost-balanced, and pair-wise independent families of hash functions. Using the family of hash functions of Dietzfelbinger~\cite{Dietzfelbinger96} will only affect multiplicative constants.

We reduce an instance of the \ThreeSUMIndexing{} problem to an instance of the \SetDisjointness{} problem as follows.
Let $R=n^{\gamma}$ for some constant $0<\gamma <1$. Let $Q = (5n/R)^2$.
Without loss of generality we assume that $\sqrt Q$ is an integer.
We pick a random hash function $h_1:U\rightarrow [R]$ from a family that is linear and almost-balanced. Using $h_1$ we create $R$ buckets $\mathcal{B}_1,\hdots, \mathcal{B}_{R}$ such that $\mathcal{B}_i=\{x\in A_1:h_1(x)=i\}$, and another $R$ buckets $\mathcal{C}_1,\hdots, \mathcal{C}_{R}$ such that $\mathcal{C}_i=\{x\in A_2:h_1(x)=i\}$.
Since $h_1$ is almost-balanced, the expected number of elements from $A_1$ and $A_2$ that are mapped to buckets of size greater than $3n/R$ is $O(R)$. We use $O(R)$ space to maintain this list explicitly, together with a lookup table for the elements in $A_1$ and $A_2$.

Next, we pick a random hash function $h_2:U\rightarrow [Q]$ where $h_2$ is chosen from a pair-wise independent and linear family. For each bucket we create $\sqrt{Q}$ shifted sets as follows:
for each $0\leq j< \sqrt Q$ let $\mathcal{B}_{i,j}= \{h_2(x)- j\cdot \sqrt{Q} \, (\modulo Q) \,|\, x\in \mathcal{B}_i\}$ and $\mathcal{C}_{i,j}= \{-h_2(x)+ j \, (\modulo Q) \,|\, x\in \mathcal{C}_i\}$.
These sets are all preprocessed into a data structure for the \SetDisjointness{} problem.

Next, we answer a \ThreeSUMIndexing{} query $z$ by utilizing the linearity of $h_1$ and $h_2$,
which implies that if there exist $x\in A_1$ and $y\in A_2$ such that $x+y=z$ then $h_1(x)+h_1(y) = h_1(z) + c_{h_1} \, (\modulo R)$ and $h_2(x)+h_2(y) = h_2(z) + c_{h_2}\, (\modulo Q)$.

Thus, if $x\in \mathcal{B}_i$ then $y$ must be in $\mathcal{C}_{h_1(z)+c_{h_1}-i (\modulo R)}$. For each $i\in [R]$ we would like to intersect $\mathcal{B}_i$ with $\mathcal{C}_{h_1(z)+c_{h_1}-i (\modulo R)}$ in order to find candidate pairs of $x$ and $y$. Denote by $h_2^{\uparrow}(z) = \lfloor \frac{h_2(z)+c_{h_1}}{\sqrt{Q}}\rfloor$ and $h_2^{\downarrow}(z) =h_2(z)+c_{h_2}(\modulo \sqrt{Q})$.
Due to the almost-linearity of $h_2$, if the sets $\mathcal{B}_i$ and $\mathcal{C}_{h_1(z) +c_{h_1}-i(\modulo R)}+z$ are not disjoint then the sets $\mathcal{B}_{i,h_2^{\uparrow}(z)}$ and $\mathcal{C}_{h_1(z)+c_{h_1}-i(\modulo R),h_2^{\downarrow}(z)}$ are not disjoint (but the reverse is not necessarily true). Thus, if $\mathcal{B}_{i,h_2^{\uparrow}(z)} \cap \mathcal{C}_{h_1(z)+c_{h_1}-i(\modulo R),h_2^{\downarrow}(z)} = \emptyset$
then there is no candidate pair in $\mathcal{B}_i$ and $\mathcal{C}_{h_1(z) +c_{h_1}-i(\modulo R)}+z$. However, if $\mathcal{B}_{i,h_2^{\uparrow}(z)} \cap \mathcal{C}_{h_1(z)+c_{h_1}-i(\modulo R),h_2^{\downarrow}(z)} \neq \emptyset$ then it is possible that this is due to a \ThreeSUMIndexing{} solution, but we may have false positives. Notice that the number of set pairs whose intersection we need to examine is $O(R)$ since $z$ is given. Once we pick $i$ ($R$ choices) the rest is implicit.

Set $z$ and let $k = h_2(z)$. Since $h_2$ is pair-wise independent and linear then for any pair $x,y\in U$ where $x\neq y$ we have that if $x+y\neq z$ then $\Pr[h_2(x) + h_2(y) = k+c_{h_2}(\modulo R)] = \Pr[h_2(x+y)= h_2(z)+c_{h_2}(\modulo R)] =\frac{1}{Q}$.
Since each bucket contains at most $3n/R$ elements, the probability of a false positive due to two buckets $\mathcal{B}_i$ and $\mathcal{C}_j$ is %$\Pr[(h_2(\mathcal{B}_i) - k -c_{h_2})\cap h_2(\mathcal{C}_j)] \leq
not greater than $(\frac{3n}{R})^2 \frac 1 {Q} = \frac 9 {25}$.
In order to reduce the probability of a false positive to be polynomially small, we repeat the process with $O(\log n)$ different choices of $h_2$ functions (but using the same $h_1$).  This blows up the number of sets by a factor of $O(\log n)$, but not the universe.
If the sets intersect under all $O(\log n)$ choices of $h_2$ then we can spend
$O(n/R)$ time to find $x$ and $y$ within buckets $\mathcal{B}_i$ and $\mathcal{C}_j$,
which are either a \ThreeSUMIndexing{} solution (and the algorithm halts),
or a false positive, which only occurs with probability $1/\mbox{poly}(n)$.

To summarize, we create a total of $O(R\sqrt{Q}\log n)$ sets, each of size at most $3n/R$. Thus, the total size of the \SetDisjointness{} instance is $N=\tilde{O}(n^2/R)$.
For a query, we perform $\tilde{O}(R)$ queries on the \SetDisjointness{} structure, and spend another $O(R \cdot \frac{n}{R} \cdot \frac{1}{poly(n)})=O(1)$ expected time to verify that we did not hit a false positive. Furthermore, we spend $O(R)$ time to check possible solutions containing one of the expected $O(R)$ elements from buckets with too many elements by using the lookup tables. If we denote by $T(N)$ and $S(N)$ the query time and space usage, respectively, of the \SetDisjointness{} data structure on $N$ elements (in our case $N= \tilde{O}(n^{2-\gamma})$), then the query time of the reduction becomes $t_{\TSI{}} = \tilde{O}(R\cdot T(n^2/R))$ time and the space usage is $s_{\TSI{}}=\tilde{O}(S(n^2/R) + O(n))$. Since we may assume that $S(N) = \Omega(N)$, we have that $s_{\TSI{}}=\tilde{O}(S(N))$.

By the Strong \ThreeSUMIndexing{} Conjecture, either $s_{\TSI{}} =\tilde \Omega (n^2)$ or $t_{\TSI{}} = \tilde \Omega (n)$, which means that either $S(N) = \tilde \Omega(N^{\frac 2 {2-\gamma}})$ or $T(N) = \tilde \Omega(N^{\frac {1-\gamma} {2-\gamma}})$. For any constant $\epsilon>0$, if the \SetDisjointness{} data structure uses $ \tilde \Theta (N^{\frac 2 {2-\gamma} -\epsilon})$ space, then $S(N)\cdot (T(N)) ^2 = \tilde \Omega (N^{\frac 2 {2-\gamma} -\epsilon + \frac{2-2\gamma}{2-\gamma}}) = \tilde \Omega (N^{2-\epsilon})$. Since this holds for any $\epsilon > 0$ it must be that $S(N)\cdot (T(N)) ^2 = \tilde \Omega (N^{2})$.
\qed
\end{proof}

\begin{theorem}\label{thm:3SI_to_SSI}
The Strong \ThreeSUMIndexing{} Conjecture implies the Strong \SetIntersection{} Conjecture.
\end{theorem}

\begin{proof}
The proof follows the same structure as the proof of Theorem~\ref{thm:3SI_to_SSD}, but here we set $Q = (n^{1+\delta}/R)$, where $\delta >0$ is a constant. Furthermore, we preprocess the buckets using a \SetIntersection{} data structure, and if two sets intersect then instead of repeating the whole process with different choices of $h_2$ (in order to reduce the probability of a false positive), we use the \SetIntersection{} data structure to report all of the elements in an intersection, and verify them all directly.

As before, set $z$ and let $k = h_2(z)$. Since $h_2$ is pair-wise independent and linear then for any pair $x,y\in U$ where $x\neq y$ we have that if $x+y\neq z$ then $\Pr[h_2(x) + h_2(y) = k+c_{h_2}(\modulo R)] = \Pr[h_2(x+y)= h_2(z)+c_{h_2}(\modulo R)] =\frac{1}{Q}$.
We now bound the expected output size from all of the intersections. Since each pair of buckets imply at most $(\frac {3n} R)^2$ pairs of elements, the expected size of their intersection is $E[|h_2(\mathcal{B}_i) - k \cap h_2(\mathcal{C}_j)|] =(\frac {3n} R)^2 \frac 1 {Q} = O(\frac{n^{1-\delta}}{R})$. Thus, the expected size of the output of all of the $O(R)$ intersections is $O(R\frac{n}{Rn^\delta}) = O(n^{1-\delta})$. For each pair in an intersection we can verify in constant time if together with $z$ they form a solution.

To summarize, we create a total of $O(R\sqrt{Q})$ sets, each of size at most $3n/R$. Thus, the total size of the \SetIntersection{} instance is $N=\tilde{O}(n^2/R)$.
For a query, we perform $\tilde{O}(R)$ queries on the \SetIntersection{} structure. Furthermore, we spend $O(R)$ time to check possible solutions containing one of the expected $O(R)$ elements from buckets with too many elements by using the lookup tables. If we denote by $T(N)$ and $S(N)$ the query time and space usage, respectively, of the \SetIntersection{} data structure on $N$ elements (in our case $N= \tilde{O}(R\sqrt{Q} n /R) = \tilde{O}(n^{\frac{3+\delta-\gamma} 2 })$), then the query time of the reduction becomes $t_{\TSI{}} = \tilde{O}(R\cdot T(N)+n^{1-\delta})$ time and the space usage is $s_{\TSI{}}=\tilde{O}(S(N) + O(n))$. Since we may assume that $S(N) = \Omega(N)$, we have that $s_{\TSI{}}=\tilde{O}(S(N))$.

By the Strong \ThreeSUMIndexing{} conjecture, either $s_{\TSI{}} =\tilde \Omega (n^2)$ or $t_{\TSI{}} = \tilde \Omega (n)$, which means that either $S(N) = \tilde \Omega(N^{\frac 4 {3+\delta-\gamma}})$ or $T(N) = \tilde \Omega(N^{\frac {2-2\gamma} {3+\delta-\gamma}})$. For any constant $\epsilon>0$, if the \SetIntersection{} data structure uses $ \tilde \Theta (N^{\frac 4 {3+\delta-\gamma} -\epsilon})$ space, then $S(N)\cdot T(N) = \tilde \Omega (N^{\frac 4 {3+\delta-\gamma} -\epsilon + \frac{2-2\gamma}{3+\delta-\gamma}}) = \tilde \Omega (N^{2 - \frac{2\delta}{3+\delta-\gamma}-\epsilon})$. Since this holds for any $\epsilon > 0$ and any $\delta >0$ it must be that $S(N)\cdot T(N) = \tilde \Omega (N^{2})$.
\qed
\end{proof}

\begin{theorem}\label{thm:WSD_to_W3SI}
The \SetDisjointness{} Conjecture implies the \ThreeSUMIndexing{} Conjecture.
\end{theorem}
\begin{proof}
Given an instance of \SetDisjointness{}, we construct an instance of \ThreeSUMIndexing{} as follows. Denote with $M$ the value of the largest element in the \SetDisjointness{} instance. Notice that we may assume that $M \leq N$ (otherwise we can use a straightforward renaming). For every element $x\in U$ that is contained in at least one of the sets we create two integers $x_A$ and $x_B$, which are represented by $2\lceil\log{m}\rceil + \lceil \log{N} \rceil + 3$ bits each (recall that $m$ is the number of sets).

The $\lceil \log{N} \rceil$ least significant bits in $x_A$ represent the value of $x$. The following bit is a zero. The following $\lceil\log{m}\rceil$ bits in $x_A$ represent the index of the set containing $x$, and the rest of the $2+\lceil\log{m}\rceil$ are all set to zero. The $\lceil \log{N} \rceil$ least significant bits in $x_B$ represent the value of $M-x$. The following $2+\lceil\log{m}\rceil$ are all set to zero. The following $\lceil\log{m}\rceil$ bits in $x_B$ represent the index of the set containing $x$, and the last bit is set to zero. Finally, the integer $x_A$ is added to $A_1$ of the \ThreeSUMIndexing{} instance, while the integer $x_B$ is added to $A_2$.

We have created two sets of $n \leq M$ integers. We then preprocess them to answer \ThreeSUMIndexing{} queries. Now, to answer a \SetDisjointness{} query on sets $S_i$ and $S_j$, we query the \ThreeSUMIndexing{} data structure with an integer $z$ which is determined as follows.
The $\lceil \log{N} \rceil$ least significant bits in $z$ represent the value of $M$. The following bit is a zero. The following $\lceil\log{m}\rceil$ bits represent the index $i$ and are followed by a zero. The next $\lceil\log{m}\rceil$ bits represent the index $j$ and the last bit is set to zero.

It is straightforward to verify that there exists a solution to the \ThreeSUMIndexing{} problem on $z$ if and only if the sets $S_i$ and $S_j$ are not disjoint.
Therefore, if there is a solution to the \ThreeSUMIndexing{} problem with less than $\tilde{\Omega}(n^2)$ space and constant query time then there is a solution for the \SetDisjointness{} problem which refutes the \SetDisjointness{} Conjecture.
\qed
\end{proof}

\section{Parameterized Generalization: \\k-Set Intersection and (k+1)-SUM}
Two parameterized generalizations of the \SetDisjointness{} and \ThreeSUMIndexing{} problems are formally defined as follows:

\begin{problem}[\KSetDisjointness{} Problem]\label{prob:k_set_dis}
Preprocess a family $F$ of $m$ sets, all from universe $U$, with total size $N=\sum_{S\in F} |S|$ so that given $k$ query sets $S_1,S_2,\dots, S_k\in F$ one can quickly
determine if $\cap_{i=1}^k S_i=\emptyset$.
\end{problem}

\begin{problem}[\KplusOneSUMIndexing{} Problem]\label{prob:kp1sum_ind}
Preprocess $k$ integer arrays $A_1, A_2, \dots, A_k$, each of length $n$, so that given a query integer $z$ we can decide if there is $x_1\in A_1, x_2\in A_2, \dots, x_k\in A_k$ such that $z=\sum_{i=1}^k x_i$.
\end{problem}

It turn out that a natural generalization of the data structure of Cohen and Porat~\cite{CP10} leads to a data structure for \KSetDisjointness{} as shown in the following lemma.
\begin{lemma}\label{lem:k_set_dis}
There exists a data structure for the
\KSetDisjointness{} problem where the query time is $T$ and the space usage is $S=O((N/T)^k)$.
\end{lemma}
\begin{proof}
We call the $f$ largest sets in $F$ \emph{large sets}. The rest of the sets are called \emph{small sets}. In the preprocessing stage we explicitly maintain a $k$-dimensional table with the answers for all \KSetDisjointness{} queries where all $k$ sets are large sets. The space needed for such a table is $S=f^k$. Moreover, for each set (large or small) we maintain a look-up table that supports disjointness queries (with this set) in constant time. Since there are $f$ large sets and the total number of elements is $N$, the size of each of the small sets is at most $N/f$.

Given a \KSetDisjointness{} query, if all of the query sets are large then we look up the answer in the $k$-dimensional table. If at least one of the sets is small then using a brute-force search we look-up each of the at most $O(N/f)$ elements in each of the other $k-1$ sets. Thus, the total query time is bounded by $O(kN/f)$, and the space usage is $S=O(f^k)$. The rest follows.
\qed
\end{proof}

Notice that for the case of $k=2$ in Lemma~\ref{lem:k_set_dis} we obtain the same tradeoff of Cohen and Porat~\cite{CP10} for \SetDisjointness{}.
The following conjecture suggests that the upper bound of Lemma~\ref{lem:k_set_dis} is the best possible.

\begin{conjecture}\label{conj:GSSIH}
\textbf{Strong \KSetDisjointness{} Conjecture}. Any data structure for the \KSetDisjointness{} problem that answers queries in $T$ time must use $S=\tilde{\Omega}(\frac{N^k}{T^k})$ space.
\end{conjecture}

Similarly, a natural generalization of the Strong \ThreeSUMIndexing{} conjecture is the following.
\begin{conjecture}
\textbf{Strong \KplusOneSUMIndexing{} Conjecture}. There is no solution for the \KplusOneSUMIndexing{} problem with $\tilde{O}(n^{k-\Omega(1)})$ space and truly sublinear query time.
\end{conjecture}

We also consider some weaker conjectures, similar to the \SetDisjointness{} and \ThreeSUMIndexing{} conjectures.

\begin{conjecture}\label{conj:GSIH}
\textbf{\KSetDisjointness{} Conjecture}. Any data structure for the \KSetDisjointness{} problem that answers queries in constant time must use $\tilde{\Omega}(N^k)$ space.
\end{conjecture}

\begin{conjecture}
\textbf{\KplusOneSUMIndexing{} Conjecture}. There is no solution for the \KplusOneSUMIndexing{} problem with $\tilde{O}(n^{k-\Omega(1)})$ space and constant query time.
\end{conjecture}

Similar to Theorem~\ref{thm:WSD_to_W3SI}, we prove the following relationship between the \KSetDisjointness{} conjecture and the \KplusOneSUMIndexing{} conjecture.

\begin{theorem}\label{thm:WKSD_to_WK+1SI}
The \KSetDisjointness{} conjecture implies the \KplusOneSUMIndexing{} conjecture
\end{theorem}
\begin{proof}

Given an instance of \KSetDisjointness{}, we construct an instance of \KplusOneSUMIndexing{} as follows.
Denote by $M$ the value of the largest element in the \SetDisjointness{} instance. Notice that we may assume that $M \leq N$ (otherwise we use a straightforward renaming). For every element $x\in U$ that is contained in at least one of the sets we create $k$ integers $x_1,x_2,...,x_k$, where each integer is represented by $k\lceil\log{m}\rceil + (k-1)\lceil\log{N}\rceil + 2k-1$ bits.

For integer $x_i$, if $i>1$ the $(k-1)\lceil \log{N} \rceil+k-1$ least significant bits are all set to zero, except for the bits in indices $(i-2)(\lceil\log{N}\rceil+1)+1,...,(i-1)(\lceil\log{N}\rceil+1)$ that represent the value of $x$. If $i=1$ the value of the bits in the indices $(j-1)(\lceil\log{N}\rceil+1)+1,...,j(\lceil\log{N}\rceil+1)$ is set to $M-x$ for all $1 \leq j \leq k-1$. The $k\lceil\log{m}\rceil+k$ following bits are all set to zero, except for the bits in indices $(i-1)(\lceil\log{m}\rceil+1)+1,...,i(\lceil\log{m}\rceil+1)$ which represent the index of the set containing $x$.

We now create an instance of \KplusOneSUMIndexing{} where the $j$th input array $A_j$ is the set of integers $x_j$ for all $x\in U$ that is contained in at least one set of our family. Thus, the size of each array is at most $N$. Now, given a \KSetDisjointness{} query $(i_1,i_2,...,i_k)$ we must decide if $S_{i_1} \cap S_{i_2} \cap ... \cap S_{i_k} = \emptyset$. To answer this query we will query the instance of \KplusOneSUMIndexing{} we have created with an integer $z$ whose binary representation is as follows: In the $(k-1)\lceil\log{N}\rceil+k-1$ least significant bits the value of the bits in the indices $(j-1)(\lceil\log{N}\rceil+1)+1,...,j(\lceil\log{N}\rceil+1)$ is set to $M$ for all $1 \leq j \leq k-1$. In the $k\lceil\log{m}\rceil+k$ following bits, the bits at locations $(j-1)(\lceil\log{m}\rceil+1)+1,...,j(\lceil\log{m}\rceil +1)$ represent $i_j$ (for $1 \leq j \leq k$). The rest of the bits are padding zero bits (in between representations of various $i_j$s and $M$s).

If $S_{i_1} \cap S_{i_2} \cap ... \cap S_{i_k} \neq \emptyset$ then by our construction it is straightforward to verify that the \KplusOneSUMIndexing{} query on $z$ will return that there is a solution. If $S_{i_1} \cap S_{i_2} \cap ... \cap S_{i_k} = \emptyset$ then at least for one $j\in[k-1]$ the sum of values in the bits in indices $(j-1)(\lceil\log{N}\rceil+1)+1,...,j(\lceil\log{N}\rceil+1)$ in the $(k-1)\lceil \log{N} \rceil+k-1$ least significant bits will not be $M$. This is because we can view each block of $\lceil\log{N}\rceil+1$ bits in the $(k-1)\lceil \log{N} \rceil+k-1$ least significant bits as solving a linear equation. This equation is of the form $M-x_1+x_i=M$ for every block $i-1$ where $2 \leq i \leq k$. The solution of each of these equations is $x_1=x_i$ for all $2 \leq i \leq k$. Consequently, a solution can be found only if there is a specific $x$ which is contained in all of the $k$ sets. Therefore, we get a correct answer to a \KSetDisjointness{} query by answering a \KplusOneSUMIndexing{} query.

Consequently, if for some specific constant $k$ there is a solution to the \KplusOneSUMIndexing{} problem with less than $\tilde{\Omega}(n^k)$ space and constant query time, then with this reduction we refute the \KSetDisjointness{} conjecture. \qed
\end{proof} 

\section{Directed Reachability Oracles as a Generalization of Set Disjointness Conjecture}

An open question which was stated by \Patrascu{} in~\cite{Patrascu11} asks if it is possible to preprocess a sparse directed graph in less than $\Omega(n^2)$ space so that \Reachability{} queries (given two query vertices $u$ and $v$ decide whether there is a path from $u$ to $v$ or not) can be answered efficiently. A partial answer, given in~\cite{Patrascu11}, states that for constant query time truly superlinear space is necessary. In the undirected case the question is trivial and one can answer queries in constant time using linear space. This is also possible for planar directed graphs (see Holm et al.~\cite{HRT14}).

We now show that \Reachability{} oracles for sparse graphs can serve as a generalization of the \SetDisjointness{} conjecture. We define the following parameterized version of \Reachability{}. In the \emph{\KReachability{} problem} the goal is to preprocess a directed sparse graph $G=(V,E)$ so that given a pair of distinct vertices $u,v \in V$ one can quickly answer whether there is a path from $u$ to $v$ consisting of at most $k$ edges. We prove that 2-\Reachability{} and \SetDisjointness{} are tightly connected.

\begin{lemma}
There is a linear time reduction from \SetDisjointness{} to 2-\Reachability{} and vice versa which preserves the size of the instance.
\end{lemma}

\begin{proof}
Given a graph $G=(V,E)$ as an instance for 2-\Reachability{}, we construct a corresponding instance of \SetDisjointness{} as follows. For each vertex $v$ we create the sets $V_{in}=\{u | (u,v) \in E\}$ and $V_{out}=\{u | (v,u) \in E\} \cup \{v\}$. We have $2n$ sets and $2m+n$ elements in all of them ($|V|=n$ and $|E|=m$). Now, a query $u,v$ is reduced to determining if the sets $U_{out}$ and $V_{in}$ are disjoint or not. Notice, that the construction is done in linear time and preserves the size of the instance. In the opposite direction, we are given $m$ sets $S_1,S_2,...,S_m$ having $N$ elements in total $e_1,e_2,...,e_N$. We can create an instance of 2-\Reachability{} in the following way. For each set $S_i$ we create a vertex $v_i$. Moreover, for each element $e_j$ we create a vertex $u_j$. Then, for each element $e_j$ in a set $s_i$ we create two directed edges $(v_i, u_j)$ and $(u_j,v_i)$. These vertices and edges define a directed graph, which is preprocessed for 2-\Reachability{} queries. It is straightforward to verify that the disjointness of $S_i$ and $S_j$ is equivalent to determining if there is a path of length at most $2$ edges from $v_i$ to $v_j$. Moreover, the construction is done in linear time and preserves the size of the instance.
\qed
\end{proof}

Furthermore, we consider \KReachability{} for $k \geq 3$. First we show an upper bound on the tradeoff between space and query time for solving \KReachability{}.

\begin{lemma}\label{lem:k_reach_UB}
There exists a data structure for \KReachability{} with $S$ space and $T$ query time  such that $ST^{2/(k-1)}=O(n^2)$.
\end{lemma}

\begin{proof}
Let $\alpha>0$ be an integer parameter to be set later. Given a directed graph $G=(V,E)$, we call vertex $v \in V$ a \emph{heavy vertex} if  $deg(v)>\alpha$  and a vertex $u \in V$a \emph{light vertex} if $deg(u)\leq \alpha$. Notice that the number of heavy vertices is at most $n/\alpha$. For all heavy vertices in $V$ we maintain a matrix containing the answers to any \KReachability{} query between two heavy vertices. This uses $O(n^2/\alpha^2)$ space.

Next, we recursively construct a data structure for \KmOneReachability{}. Given a query $u,v$, if both vertices are heavy then the answer is obtained from the matrix. Otherwise, either $u$ or $v$ is light vertex. Without loss of generality, say $u$ is a light vertex. We consider each vertex $w \in N_{out}(u)$ ($N_{out}(u)=\{v| (u,v) \in E \}$) and query the \KmOneReachability{} data structure with the pair $w,v$. Since $u$ is a light node, there are no more than $\alpha$ queries. One of the queries returns a positive answer if and only if there exists a path of length at most $k$ from $u$ to $v$.

Denote by $S(k,n)$ the space used by our \KReachability{} oracle on a graph with $n$ vertices and denote by $Q(k,n)$ the corresponding query time. In our construction we have $S(k,n) = n^2/ \alpha^2 + S(k-1,n)$ and $Q(k,n) = \alpha Q(k-1,n)+O(1)$. For $k=1$ it is easy to construct a linear space data structure using hashing so that queries can be answered in constant time. Thus, $S = S(k,n) = O((k-1)n^2/ \alpha^2)$ and $T = Q(k,n)= O(\alpha^{k-1})$.
\qed
\end{proof}

Notice that for the case of $k=2$ the upper bounds from Lemma~\ref{lem:k_reach_UB} exactly match the tradeoff of the Strong \SetDisjointness{} Conjecture ($ST^2=\tilde{O}(n^2)$). We expand this conjecture by considering the tightness of our upper bound for \KReachability{}, which then leads to some interesting consequences with regard to distance oracles.

\begin{conjecture}\label{DCRH}
\textbf{Directed \KReachability{} Conjecture}. Any data structure for the \KReachability{} problem with query time $T$ must use $S =\tilde{\Omega}(\frac{n^2}{T^{2/(k-1)}})$ space.
\end{conjecture}

Notice that when $k$ is non-constant then by our upper bound $\tilde{\Omega}(n^2)$ space is necessary independent of the query time. This fits nicely with what is currently known about the general question of \Reachability{} oracles: either we spend $n^2$ space and answer queries in constant time or we do no preprocessing and then answer queries in linear time. This leads to the following conjecture.

\begin{conjecture}\label{DRH}
\textbf{Directed Reachability Hypothesis}. Any data structure for the \Reachability{} problem must either use $\tilde \Omega(n^2)$ space, or linear query time.
\end{conjecture}

The conjecture states that in the general case of \Reachability{} there is no full tradeoff between space and query time. We believe the conjecture is true even if the path is limited to lengths of some non-constant number of edges. 

\section{Distance Oracles and Directed Reachability}

There are known lower bounds for constant query time distance oracles based on the \SetDisjointness{} hypothesis. Specifically, Cohen and Porat~\cite{CP10} showed that stretch-less-than-2 oracles need $\Omega(n^2)$ space for constant queries. Patrascu et al.~\cite{PRT12} showed a conditional space lower bound of $\Omega(m^{5/3})$ for constant-time stretch-2 oracles. Applying the Strong \SetDisjointness{} conjecture to the same argument as in~\cite{CP10} we can prove that for stretch-less-than-2 oracles the tradeoff between $S$ (the space for the oracle) and $T$ (the query time) is by $S \times T^2 = \Omega(n^2)$.

Recent effort was taken toward constructing compact distance oracles where we allow non-constant query time. For stretch-2 and stretch-3 Agarwal et al.~\cite{AGH11}~\cite{AGH12} achieves a space-time tradeoff of $S \times T = O(n^2)$ and $S \times T^2 = O(n^2)$, respectively, for sparse graphs. Agarwal~\cite{Agarwal14} also showed many other results for stretch-2 and below. Specifically, Agarwal showed that for any integer $k$ a stretch-(1+1/k) oracle exhibits the following space-time tradeoff: $S \times T^{1/k} = O(n^2)$. Agarwal also showed a stretch-(1+1/(k+0.5)) oracle that exhibits the following tradeoff: $S \times T^{1/(k+1)} = O(n^2)$. Finally, Agarwal gave a stretch-(5/3) oracle that achieves a space-time tradeoff of $S \times T = O(n^2)$. Unfortunately, no lower bounds are known for non-constant query time.

Conditioned on the directed \KReachability{} conjecture we prove the following lower bound.

\begin{lemma}
Assume the directed \KReachability{} conjecture holds.
Then stretch-less-than-$(1+2/k)$ distance oracles with query time $T$ must use $S \times T^{2/(k-1)} = \tilde{\Omega}(n^2)$ space.
\end{lemma}

\begin{proof}
Given a graph $G=(V,E)$ for which we want to preprocess for \KReachability{}, we create a layered graph with $k$ layers where each layer consists of a copy of all vertices of $V$. Each pair of neighboring layers is connected by a copy of all edges in $E$. We omit all directions from the edges. For every fixed integer $k$, the layered graph has $O(|V|)$ vertices and $O(|E|)$ edges. Next, notice that if we construct a distance oracle that can distinguish between pairs of vertices of distance at most $k$ and pairs of vertices of distance at least $k+2$, then we can answer \KReachability{} queries. Consequently, assuming the \KReachability{} conjecture we have that $S \times T^{2/(k-1)} = \Omega(n^2)$ for stretch-less-than-$(1+2/k)$ distance oracles (For $k=2$ this is exactly the result we get by the \SetDisjointness{} hypothesis).
\qed
\end{proof}

Notice, that the stretch-(5/3) oracle shown by Agarwal~\cite{Agarwal14} achieves a space-time tradeoff of $S \times T = O(n^2)$. Our lower bound is very close to this upper bound since it applies for any distance oracle with stretch-less-than-$(5/3)$, by setting $k=3$.

\section{SETH and Orthogonal Vectors Space Conjectures}\label{sec:SETH_OV}

Solving SAT using $O(2^n)$ time where $n$ is number of variables in the formula can be easily done using only $O(n)$ space. However, the question is how can we use space in the case that we have only a partial assignment of $R$ variables and we would like to quickly figure out whether this partial assignment can be completed to a full satisfying assignment or not. On one end, by using just $O(n)$ space we can answer queries in $O(2^{n-R})$ time. On the other end, we can save the answers to all possible queries using $O(2^R)$ space. It is not clear if there is some sort of a tradeoff in between these two. A related problem is the problem of Orthogonal Vectors (OV). In this problem one is given a collection of $n$ vectors of length $O(\log{n})$ and need to answer if there are two of them which are orthogonal to one another. A reduction from SETH to OV was shown in~\cite{Williams05}. By this reduction given a k-CNF formula of $n$ variables one can transform it using $O(2^{\epsilon n})$ time to $O(2^{\epsilon n})$ instances of OV in which the vectors are of length $2f(k,\epsilon)\log{n}$ (for any $\epsilon > 0$, where $f(k,\epsilon)n$ is the number of clauses of each sparse formula represented by one instance of OV). This reduction leads to the following conjecture regarding OV, which is based on SETH: There is no algorithm that, for every $c \geq 1$, solves the OV problem on $n$ boolean vectors of length $c\log{n}$ in $\tilde{O}(n^{2-\Omega(1)})$ time.

We can consider a data structure variant of the OV problem, which we call OV indexing. Given a list of $n$ boolean vectors of length $c\log{n}$ we should preprocess them and create a suitable data structure. Then, we answer queries of the following form: Given a vector $v$, is there a vector in the list which is orthogonal to $v$?

We state the following conjecture which is the space variant of the well-studied OV (time) conjecture:
\begin{conjecture}
\textbf{Orthogonal Vectors Indexing Hypothesis}: There is no algorithm for every $c \geq 1$  that solves the OV indexing problem with $\tilde{O}(n^{2 -\Omega(1)})$ space and truly sublinear query time.
\end{conjecture}

We note that we believe that the last conjecture is true even if we allow superpolynomial preprocessing time. Moreover, it seems that it also may be true even for some constant $c$ slightly larger than 2. 

\section{Space Requirements for Boolean Matrix Multiplication}

Boolean Matrix Multiplication(BMM) is one of the most fundamental problems in Theoretical Computer Science. The question of whether computing the Boolean product of two Boolean matrices of size $ n \times n$ is possible in $O(n^2)$ time is one of the most intriguing open problems. Moreover, finding a \emph{combinatorial} algorithm for BMM taking $O(n^{3-\epsilon})$ time for some $\epsilon >0$ is considered to be impossible to do with current algorithmic techniques.

We focus on the following data structure version of BMM, preprocess two $n \times n$ Boolean matrices $A$ and $B$, such that given a query $(i,j)$ we can quickly return the value of $c_{i,j}$ where $C=\{c_{i,j}\}$ is the Boolean produce $A$ and $B$.
Since storing all possible answers to queries will require $\theta(n^2)$ space in the worst case, we focus on the more interesting scenario where we have only $O(n^{2-\Omega(1)})$ space to store the outcome of the preprocessing stage. In case the input matrices are dense (the number of ones and the number of zeroes are both $\theta(n^2)$) it seems that this can be hard to achieve as storing the input matrices alone will take $\theta(n^2)$ space. So we consider a complexity model, which we call the \emph{read-only input model},  in which storing the input is for free (say on read-only memory), and the space usage of the data structure is only related to the additional space used. We now demonstrate that BMM in the read-only input model is equivalent to \SetDisjointness{}.

\begin{lemma}
BMM in the read-only input model and \SetDisjointness{} are equivalent.
\end{lemma}

\begin{proof}
Given an instance of \SetDisjointness{} let $e_1,...,e_N$ denote the elements in an input instance. We construct an instance of BMM as follows. Assume without loss of generality that all sets are not empty, and so $m \leq N$. Row $i$ in matrix $A$ represents a set $S_i$ while each column $j$ represents element $e_j$. An entry $a_{i,j}$ equals $1$ if $e_j \in S_i$ and equals zero otherwise. We also set $B = A^T$. We also pad each of the matrices with zeroes so their size will be $N \times N$. Clearly, $c_{i,j}$ in matrix $C$, which is the product of $A$ and $B$, is an indicator whether $S_i \cap S_j=\emptyset$.

In the opposite direction, given two matrices $A$ and $B$ having $m$ ones we view each row $i$ of $A$ as a characteristic vector of a set $S_i$ (the elements in the set correspond to the ones in that row) and each column $j$ of $B$ as a characteristic vector of a set $S_{j+n}$ (the elements in the set corresponds to the ones in that column). Thus, the instance of \SetDisjointness{} that have been created consists of $2n$ set with $O(m)$ elements. The value of an element $c_{i,j}$ in the product of $A$ and $B$ can be determined by the intersection of $S_i$ and $S_{j+n}$.
\qed
\end{proof}

Another interesting connection between BMM and the other problems discussed in this paper is the connection to the problem of calculating the \emph{transitive closure} of a graph, which is the general directed reachability mentioned above. It is well-known that BMM and transitive closure are equivalent in terms of time as shown by Fischer and Meyer~\cite{FM71}. But what happens if we consider space? It is easy to see that BMM can be reduced to transitive-closure (directed reachability) even in terms of space. However, the opposite direction is not clear as the reduction for time involves recursive squaring, which cannot be implemented efficiently in terms of space.

Another fascinating variant of BMM is the one in which an $n \times n$ matrix $A$ is input for preprocessing and afterwards we need to calculate the result of multiplying it by a given query vector $v$. This can be seen as the space variant of the celebrated OMV (online matrix-vector) problem discussed by Henzinger et al.~\cite{HKNS15}. It is interesting to see if one can make use of a data structure so that $n$ consecutive vector queries can be answered in $\tilde{O}(n^{3-\Omega(1)})$ time. 

\section{Applications}

We now provide applications of our rich framework for proving conditional space lower bounds. In the following subsections we consider both static and dynamic problems.

\subsection{Static Problems}

\subsubsection{Edge Triangles}

The first example we consider is in regards to triangles. In a problem that is called \emph{edge triangles detection}, we are given a graph $G=(V,E)$ to preprocess and then we are given an edge $(v,u)$ as a query and need to answer whether $(u,v)$ belongs to a triangle. In a reporting variant of this problem, called \emph{edge triangles} we need not only to answer if $(u,v)$ belongs to a triangle but also report all triangles it belongs to. This problem was considered in~\cite{AKLPPS15}.

It can be easily shown that these problems are equivalent to \SetDisjointness{} and \SetIntersection{}. We just construct a set $S_v$ per each vertex $v$ containing all its neighbors. Querying if there is a triangle containing the edge $(u,v)$ is equivalent to asking if $S_v \cap S_u$ is empty or not. Considering the reporting variant, reporting all triangles containing $(u,v)$ is thus equivalent to finding all the elements in $S_v \cap S_u$. Therefore, we get the following results:

\begin{theorem}
Assume the Strong \SetDisjointness{} conjecture.
Suppose there is a data structure for edge triangles detection problem for a graph $G=(V,E)$, with $S$ space and query time $T$. Then $S=\tilde \Omega(|E|^2/T^2)$.
\end{theorem}

\begin{theorem}
Assume the Strong \SetIntersection{} conjecture.
Suppose there is a data structure for edge triangles problem for a graph $G=(V,E)$, with $S$ space and query time $O(T + op)$ time, where $op$ is the size of the output of the query. Then $S=\tilde \Omega(|E|^2/T)$.
\end{theorem}

\subsubsection{Histogram Indexing}
A {\em histogram}, also called a {\em Parikh vector}, of a string $T$ over alphabet $\Sigma$ is a $|\Sigma|$-length vector containing the character count of $T$. For example, for $T=aaccbacab$ the histogram is $v(T)=(4,2,3)$.
In the \emph{histogram indexing problem} we preprocess an $N$-length string $T$ to support the following queries: given a query histogram $v$, return whether there is a substring $T'$ of $T$ such that $v(T')=v$.

This problem has received much attention in the recent years. The case where the alphabet size is 2 (binary alphabet) was especially studied. A simple algorithm for this case solves the problem in $O(N^2)$ preprocessing time and constant query time. There was a concentrated effort to reduce the quadratic preprocessing time for some years. However, an algorithm with preprocessing time that is $O(N^{2-\epsilon})$ for some $\epsilon>0$ was unknown until a recent breakthrough by Chan and Lewenstein~\cite{CL15}. They showed an algorithm with $O(N^{1.859})$ preprocessing time and constant query time. For alphabet size $\ell$ they obtained an algorithm with $\tilde{O}(N^{2-\delta})$ preprocessing time and $\tilde{O}(N^{2/3+\delta(\ell+13)/6})$ query time for $0 \leq \delta \leq 1$. Regarding space complexity, it is well known how to solve histogram indexing for binary alphabet using linear space and constant query time. For alphabet size $\ell$, Kociumaka et al.~\cite{KRR13} presented a data structure with $\tilde{O}(N^{2-\delta})$ space and $\tilde{O}(N^{\delta (2\ell-1)})$ query time. Chan and Lewenstein~\cite{CL15} improved their result and showed a solution by a data structure using $\tilde{O}(N^{2-\delta})$ space with only $\tilde{O}(N^{\delta(\ell +1)/2})$ query time.

Amir et al.~\cite{ACLL14} proved conditional lower bound on the tradeoff between the preprocessing and query time of the histogram indexing problem. Very recently, their lower bound was improved and generalized by Goldstein et al.~\cite{GKLP16}.
Following the reduction by Goldstein et al.~\cite{GKLP16} and utilizing our framework for conditional space lower bounds, we obtain the following lower bound on the tradeoff between the space and query time of histogram indexing:

\begin{theorem}\label{thm:generalized_space_histogram_indexing}
Assume the Strong \ThreeSUMIndexing{} conjecture holds.
The histogram indexing problem for a string of length $N$ and constant alphabet size $\ell \geq 3$ cannot be solved with $O(N^{2-\frac{2(1-\alpha)}{\ell-1-\alpha}-\Omega(1)})$ space and $O(N^{1-\frac{1+\alpha(\ell-3)}{\ell-1-\alpha}-\Omega(1)})$ query time, for any $0 \leq \alpha \leq 1$.
\end{theorem}

\begin{proof}
We use the same reduction as in \cite{GKLP16}. This time it will be used to reduce an instance of \ThreeSUMIndexing{} (on $2n$ numbers) to histogram indexing, instead of reducing from an instance of \ThreeSUM{}. The space consumed by the reduction is dominated by the space needed to prepare a histogram indexing instance with string length $N=O(n^{\frac{\ell-2-\alpha}{\ell-3}})$ for histogram queries. The number of histogram queries we do for each query number $z$ of the \ThreeSUMIndexing{} instance is $O(n^{\alpha})$. The query time is dominated by the time required by these queries. Let $S(N,\ell)$ denote the space required by a data structure for histogram indexing on $N$-length string over alphabet size $\ell$ and let $Q(N,\ell)$ denote the query time for the same parameters. Assuming the strong \ThreeSUMIndexing{} conjecture and following our reduction, we have that $S(N,\ell)=O(n^{2-\Omega(1)})$ and $Q(N,\ell)=O(n^{1-\alpha-\Omega(1)})$. Plugging in the value of $n$ in terms of $N$ we get the required lower bound.
\qed
\end{proof}

If we plug in the previous theorem $\delta=\frac{2(1- \alpha)}{\ell - 1 - \alpha}$, we get that if the strong \ThreeSUMIndexing{} conjecture is true we cannot have a solution for histogram indexing with $\tilde{O}(N^{2-\delta})$ space and $\tilde{O}(N^{\delta(\ell-2)/2})$ query time. This lower bound is very close to the upper bound obtained by Chan and Lewenstein~\cite{CL15} as there is only a gap of $\frac{3}{2}\delta$ in the power of $N$ in the query time. Moreover, if the value of $\delta$ becomes close to 0 (so the value of $\alpha$ is close to 1) the upper bound and the lower bound get even closer to each other. This is very interesting, as it means that to get truly subquadratic space solution for histogram indexing for alphabet size greater than 2, we will have to spend polynomial query time. This is in stark contrast to the simple linear space solution for histogram indexing over binary alphabets that supports queries in constant time.

\bigskip

Following reductions presented in~\cite{KPP16}, from \SetIntersection{} or \SetDisjointness{} to several other problems, we are able to show that based on the Strong \SetDisjointness{} conjecture, the same problems admit a space/query time lower bounds. For sake of completeness, we reproduce these reductions in the next three subsections and show that they admit the space lower bounds as needed.

\subsubsection{Distance Oracles for Colors}\label{section:colored_dist}% - Theorem~\ref{thm:dist_oracle_color}

Let $P$ be a set of points in some metric with distance function $d(\cdot,\cdot)$, where each point $p\in P$ has some associated colors $C(p)\subset [\ell]$. For $c\in [\ell]$ we denote by $P(c)$ the set of points from $P$ with color $c$. We generalize $d$ so that the distance between a point $p$ and a color $c$ is denoted by $d(p,c)=\min_{q\in P(c)}\{d(p,q)\}$. In the \textit{(Approximate) Distance Oracles for Vertex-Labeled Graphs problem}~\cite{Chechik12,HLWY11} we are interested in preprocessing $P$ so that given a query of a point $q$ and a color $c$ we can return $d(q,c)$ (or some approximation). We further generalize $d$ so that the distance between two colors $c$ and $c'$ is denoted by $d(c,c') = \min_{p\in P(c)}\{d(p,c')\}$. In the \textit{Distance Oracle for Colors problem} we are interested in preprocessing $P$ so that given two query colors $c$ and $c'$ we can return $d(c,c')$. In the \textit{Approximate Distance Oracle for Colors problem} we are interested in preprocessing $P$ and some constant $\alpha > 1 $ so that given two query colors $c$ and $c'$ we can return some value $\hat d$ such that $d(c,c') \leq \hat d \leq \alpha d(c,c')$.

We show evidence of the hardness of the Distance Oracle for Colors problem and the Approximate Distance Oracle for Colors problem by focusing on the 1-D case.
\begin{theorem}\label{thm:dist_oracle_color}
Assume the Strong \SetDisjointness{} conjecture.
Suppose there is a 1-D Distance Oracle for Colors with constant stretch $\alpha \geq 1$ for an input array of size $N$ with $S$ space and query time $T$. Then $S=\tilde \Omega(N^2/T^2)$.
\end{theorem}

\begin{proof}
We reduce \SetDisjointness{} to the Colored Distance problem as follows. For each set $S_i$ we define a unique color $c_i$. For an element $e\in U$ ($U$ is the universe of the elements in our sets) let $|e|$ denote the number of sets containing $e$ and notice that $\sum_{e\in U} |e| =N$. Since each element in $U$ appears in at most $m$ sets, we partition $U$ into $\Theta(\log m)$ parts where the $i^{th}$ part $P_i$ contains all of the elements $e\in U$ such that $2^{i-1}<|e| \leq 2^i$. An array $X_i$ is constructed from $P_i=\{e_1,\cdots e_{|P_i|}\}$ by assigning an interval $I_j=[f_j, \ell_j]$ in $X_i$ to each $e_j\in P_i$ such that no two intervals overlap. Every interval $I_j$ contains all the colors of sets that contain $e_j$. This implies that $|I_j| = |e_j| \leq 2^i$. Furthermore, for each $e_j$ and $e_{j+1}$ we separate $I_j$ from $I_{j+1}$ with a dummy color $d$ listed $2^i +1$ times at locations $[\ell_j+1,f_{j+1}-1]$. 

We can now simulate a \SetDisjointness{} query on subsets $(S_i,S_j)$ by performing a colored distance query on colors $c_i$ and $c_j$ in each of the $\Theta(\log m)$ arrays. There exists a $P_i$ for which the two points returned from the query are at distance strictly less than $2^i+1$ if and only if there is an element in $U$ that is contained in both $S_i$ and $S_j$. The space usage is $\tilde{O}(S)$ and the query time is $\tilde{O}(T)$. The rest follows directly from the Strong \SetDisjointness{} conjecture.

Finally, notice that the lower bound also holds for the approximate case, as for any constant $\alpha$ the reduction can overcome the $\alpha$ approximation by separating intervals using $\alpha2^i + 1$ listings of $d$.
\qed
\end{proof}

\subsubsection{Document Retrieval Problems with Multiple Patterns}

In the \textit{Document Retrieval problem}~\cite{Muthukrishnan02} we are interested in preprocessing a collection of documents $X=\{D_1,\cdots, D_k\}$ where $N=\sum_{D\in X} |D|$, so that given a pattern $P$ we can quickly report all of the documents that contain $P$. Typically, we are interested in run time that depends on the number of documents that contain $P$ and not in the total number of occurrences of $P$ in the entire collection of documents. In the \textit{Two Patterns Document Retrieval problem} we are given two patterns $P_1$ and $P_2$ during query time, and wish to report all of the documents that contain both $P_1$ and $P_2$.
We consider two versions of the Two Patterns Document Retrieval problem. In the decision version we are only interested in detecting if there exists a document that contains both patterns. In the reporting version we are interested in enumerating all documents that contain both patterns.

All known solutions for the Two Patterns Document Retrieval problem with non trivial preprocessing use at least $\Omega(\sqrt N)$ time per query~\cite{CP10,HSTV10,HSTV12,Muthukrishnan02}. In a recent paper, Larsen, Munro, Nielsen, and Thankachan~\cite{LMNT14} show lower bounds for the Two Patterns Document Retrieval problem conditioned on the hardness of boolean matrix multiplication.

It is straightforward to see that the appropriate versions of the two pattern document retrieval problem solve the corresponding versions of the \SetDisjointness{} and \SetIntersection{} problems. In particular, this can be obtained by creating an alphabet $\Sigma=F$ (one character for each set), and for each $e\in U$ we create a document that contains the characters corresponding to the sets that contain $e$. The intersection between $S_i$ and $S_j$ directly corresponds to all the documents that contain both $a$ and $b$. Thus, all of the lower bound tradeoffs for intersection problems are lower bound tradeoffs for the two pattern document retrieval problem.

\begin{theorem}\label{thm:two_pat_doc_retrieval}
Assume the Strong \SetDisjointness{} conjecture.
Suppose there is a data structure for the decision version of the Two Patterns Document Retrieval problem for a collection of documents $X$ where $N=\sum_{D\in X}|D|$, with $S$ space and query time $T$. Then $S=\tilde \Omega(N^2/T^2)$.
\end{theorem}

\begin{theorem}\label{thm:two_pat_doc_retrieval_reporting}
Assume the Strong \SetIntersection{} conjecture.
Suppose there is a data structure for the reporting version of the Two Patterns Document Retrieval problem for a collection of documents $X$ where $N=\sum_{D\in X}|D|$, with $S$ space and query time $O(T+op)$ where $op$ is the size of the output. Then $S=\tilde \Omega(N^2/T)$.
\end{theorem}

\subsubsection{Forbidden Pattern Document Retrieval}\label{section:forbidden_pattern_document_retrieval}

In the \textit{Forbidden Pattern Document Retrieval problem}~\cite{FGKLMSV12} we are also interested in preprocessing the collection of documents but this time given a query $P^+$ and $P^-$ we are interested in reporting all of the documents that contain $P^+$ and do not contain $P^-$. Here too we consider a decision version and a reporting version.

All known solutions for the Forbidden Pattern Document Retrieval problem with non trivial preprocessing use at least $\Omega(\sqrt N)$ time per query~\cite{FGKLMSV12,HSTV12}. In a recent paper, Larsen, Munro, Nielsen, and Thankachan~\cite{LMNT14} show lower bounds for the Forbidden Pattern Document Retrieval problem conditioned on the hardness of boolean matrix multiplication.

\begin{theorem}\label{thm:forbid_pat_doc_retrieval}
Assume the Strong \ThreeSUMIndexing{} conjecture.
Suppose there is a data structure for the decision version of the Forbidden Pattern Document Retrieval problem for a collection of documents $X$ where $N=\sum_{D\in X}|D|$, with $S$ space and query time $T$. Then $S=\tilde \Omega(N^2/T^4)$.
\end{theorem}

\begin{proof}
We will make use of the hard instance of \SetDisjointness{} that was used in order to prove Theorem~\ref{thm:3SI_to_SSD}, and reduce this specific hard instance to the decision version of the Forbidden Pattern Document Retrieval problem. Recall that the size of this hard instance is $\tilde{O}(n^{2-\gamma})$, the universe size is $O(n^{2-2\gamma})$, the number of sets is $\tilde{O}(n)$, and we need to perform $\tilde{O}(n^\gamma)$ \SetDisjointness{} queries in order to answer one \ThreeSUMIndexing{} query.

Similar to the proof of Theorem~\ref{thm:two_pat_doc_retrieval} we set $\Sigma=F$ (one character for each set). However, this time for each $e$ we create a document that contains all the characters corresponding to sets $\mathcal{B}_{i,j}$ that contain $e$ and all the characters corresponding to sets $\mathcal{C}_{i,j}$ that do not contain $e$.

The reason that we prove our lower bound based on the Strong \ThreeSUMIndexing{} conjecture and not on the Strong \SetDisjointness{} conjecture is because the size of our instance can become rather large relative to $N$ (as the number of sets that do not contain an element can be extremely large).

Thus, the size of the Forbidden Pattern Document Retrieval instance is $N = \theta(n^{3-2\gamma})$, and the number of queries to answer is $\theta(n^{\gamma} )$. %= \Theta(n^{1+\gamma}) = \Theta(N^{\frac{1+\gamma}{3-2\gamma}})$.
Notice that the size of the instance enforces $\gamma$ to be strictly larger than $1/2$.
By the Strong \ThreeSUMIndexing{} conjecture, either $S = s_{\TSI{}} =\tilde \Omega (n^2) = \tilde \Omega(N^{\frac 2 {3-2\gamma}})$ or $O(n^\gamma T) \geq  t_{\TSI{}} \geq \tilde \Omega (n) $, and so $T \geq \tilde \Omega(N^{\frac {1-\gamma} {3-2\gamma}})$.
For any constant $\epsilon>0$, if the Forbidden Pattern Document Retrieval data structure uses $ \tilde \Theta (N^{\frac 2 {3-2\gamma} -\epsilon})$ space, then $S \cdot T^4 = \tilde \Omega (N^{\frac 2 {3-2\gamma} -\epsilon + \frac{4-4\gamma}{3-2\gamma}}) = \tilde \Omega (N^{2-\epsilon})$. Since this holds for any $\epsilon > 0$ it must be that $S \cdot T^4 = \tilde \Omega (N^{2})$.
\qed

\end{proof}

Notice that
if we only allow linear space then we obtain a query time lower bound of $\Omega(N^{\frac{1}{4}-o(1)})$.

\begin{theorem}\label{thm:forbid_pat_doc_retrieval_reporting}
Assume the Strong \ThreeSUMIndexing{} conjecture.
Suppose there is a data structure for the reporting version of the Forbidden Pattern Document Retrieval problem for a collection of documents $X$ where $N=\sum_{D\in X}|D|$, with $S$ space and query time $O(T+op)$ where $op$ is the size of the output.  Then $S=\tilde \Omega(N^2/T)$.
\end{theorem}

\begin{proof}
Our proof is similar to the proof of Theorem~\ref{thm:forbid_pat_doc_retrieval}, only this time we use the hard instance of \SetIntersection{} from Theorem~\ref{thm:3SI_to_SSI}.
So the number of queries is $\tilde{O}(n^{\gamma})$, the size of the universe is $O(n^{1+\delta-\gamma})$, the number of sets is $\tilde{O}(n^{\frac{1+\delta+\gamma}2})$, and the total size of the output is $\Theta(n^{2-\delta})$.

Thus, the size of the Forbidden Pattern Document Retrieval instance is $N=\Theta(n^{1+\delta-\gamma}n^{\frac{1+\delta+\gamma}2}) = \Theta(n^{\frac{3+3\delta-\gamma}{2}})$, and the number of queries to answer is $\theta(n^{\gamma} )$.
Notice that the size of the instance enforces $3\delta- \gamma  < 1$.
By the Strong \ThreeSUMIndexing{} conjecture, either $S = s_{\TSI{}} =\tilde \Omega (n^2) = \tilde \Omega(N^{\frac 4 {3+3\delta -\gamma}})$ or $O(n^\gamma T) \geq  t_{\TSI{}} \geq \tilde \Omega (n) $, and so $T \geq \tilde \Omega(N^{\frac {2-2\gamma} {3+3\delta -\gamma}})$.
For any constant $\epsilon>0$, if the Forbidden Pattern Document Retrieval data structure uses $ \tilde \Theta (N^{\frac 4 {3+3\delta -\gamma} -\epsilon})$ space, then $S \cdot T = \tilde \Omega (N^{\frac 4 {3+3\delta -\gamma} -\epsilon + \frac{2-2\gamma}{{3+3\delta -\gamma}}}) = \tilde \Omega (N^{2-\frac{6\delta}{3+3\delta-\gamma}-\epsilon})$. Since this holds for any $\epsilon > 0$ and since we can make $\frac{6\delta}{3+3\delta-\gamma}$ as small as we like, it must be that $S \cdot T = \tilde \Omega (N^{2})$.
\qed

\end{proof}

\subsection{Dynamic Problems}

We show space lower bounds on dynamic problems. Lower bounds for these problems from the time perspective were considered by Abboud and Vassilevska-Williams~\cite{AW14}. The first dynamic problem we consider is \emph{st-SubConn} which is defined as follows. Given an undirected graph $G=(V,E)$, two fixed vertices $s$ and $t$ and a set $S \subseteq V$, answer whether $s$ and $t$ are connected using vertices form $S$ only. Vertices can be added or removed from $S$.

The \SetDisjointness{} problem can be reduced to st-SubConn. Given an instance of \SetDisjointness{} we create an undirected graph $G=(V,E)$ as follows. We first create two unique vertices $s$ and $t$. Then, for each set $S_i$ we create two vertices $v_i$ and $u_i$ and for each element $e_j$ we create a vertex $w_j$. Moreover, we define $E = \{ (v_i, w_j) | e_j \in S_i \} \cup \{ (u_i, w_j) | e_j \in S_i \} \cup \{(s,v_i)| 1 \leq i \leq m\} \cup \{(u_i,t)| 1 \leq i \leq m\}$. Initially the set $S$ contains $s$ and $t$ and all the $w_i$s. Given a query $(i,j)$ asking about the emptiness of $S_i \cap S_j$, we add $v_i$ and $u_j$ to the set $S$. Then, we ask if $s$ and $t$ are connected, if so we know that $S_i \cap S_j$ is not empty as the only way to get from $s$ to $t$ is following $v_i$ and $u_j$ and some node representing a common element of $S_i$ and $S_j$. If $s$ and $t$ are not connected then it is clear that the intersection is empty. After the query we remove the two vertices we have added so other queries can be handled properly.
By this construction we get the following result:

\begin{theorem}
Assume the Strong \SetDisjointness{} conjecture.
Suppose there is a data structure for st-SubConn problem for a graph $G=(V,E)$, with $S$ space and update and query time $T$. Then $S=\tilde \Omega(|E|^2/T^2)$.
\end{theorem}

There are other dynamic problems that st-SubConn can be efficiently reduced to, as shown by Abboud and Vssilevska-Williams~\cite{AW14}. This includes the following 3 problems:

\begin{problem}
\textbf{(s,t)-Reachability (st-Reach)}. Maintain a directed graph $G=(V,E)$ subject to edge insertions and deletions, so that queries about the reachability of fixed vertices $s$ and $t$ can be answered quickly.
\end{problem}

\begin{problem}
\textbf{Bipartite Perfect Matching (BPMatch)}. Preprocess and maintain undirected bipartite graph $G=(V,E)$ subject to edge insertions and deletions, so that we can quickly answer if the graph has perfect matching.
\end{problem}

\begin{problem}
\textbf{Strong Connectivity (SC)}. Preprocess and maintain directed graph $G=(V,E)$ subject to edge insertions and deletions, so that we can quickly answer if the graph is strongly connected
\end{problem}

Using our last theorem and the reductions by Abboud and Vassilevska-Williams~\cite{AW14}, noting that they do not effect the space usage, we get the following:

\begin{theorem}
Assume the Strong \SetDisjointness{} conjecture.
Suppose there is a data structure for st-Reach/ BPMatch/ SC problem for a graph $G=(V,E)$, with $S$ space and update and query time $T$. Then $S=\tilde \Omega(|E|^2/T^2)$.
\end{theorem}

We can get better lower bound for these 3 problems on sparse graphs based on the directed reachability conjecture. Given a sparse graph $G=(V,E)$ as an instance of directed reachability we can reduce it to an instance of st-Reach by just adding to special nodes $s$ and $t$ to the graph. Then, we can answer queries of the form "Is $v$ reachable from $u$?" by inserting two edges $(s,v)$ and $(u,t)$ and asking if $t$ is reachable from $s$. After the query we can restore the initial state by deleting these two edges. Thus, by using the reductions from st-Reach to BPMatch and SC as shown in ~\cite{AW14}, we get the following hardness result:

\begin{theorem}
Assume the Directed Reachability conjecture.
Any data structure for the st-Reach/ BPMatch/ SC problem on sparse graphs can not have $\tilde{O}(n^{1-\Omega(1)})$ update and query time and $\tilde{O}(n^{2-\Omega(1)})$ space.
\end{theorem}

\bibliographystyle{plain} \bibliography{ms}

\newpage
\appendix
\section*{Appendix}
\section{Sketch of the Main Results}\label{sec:sketch}

\begin{figure}[ht]\label{fig:connections}
\centering
\begin{adjustbox}{width=1.2\textwidth,center=18cm}
\input{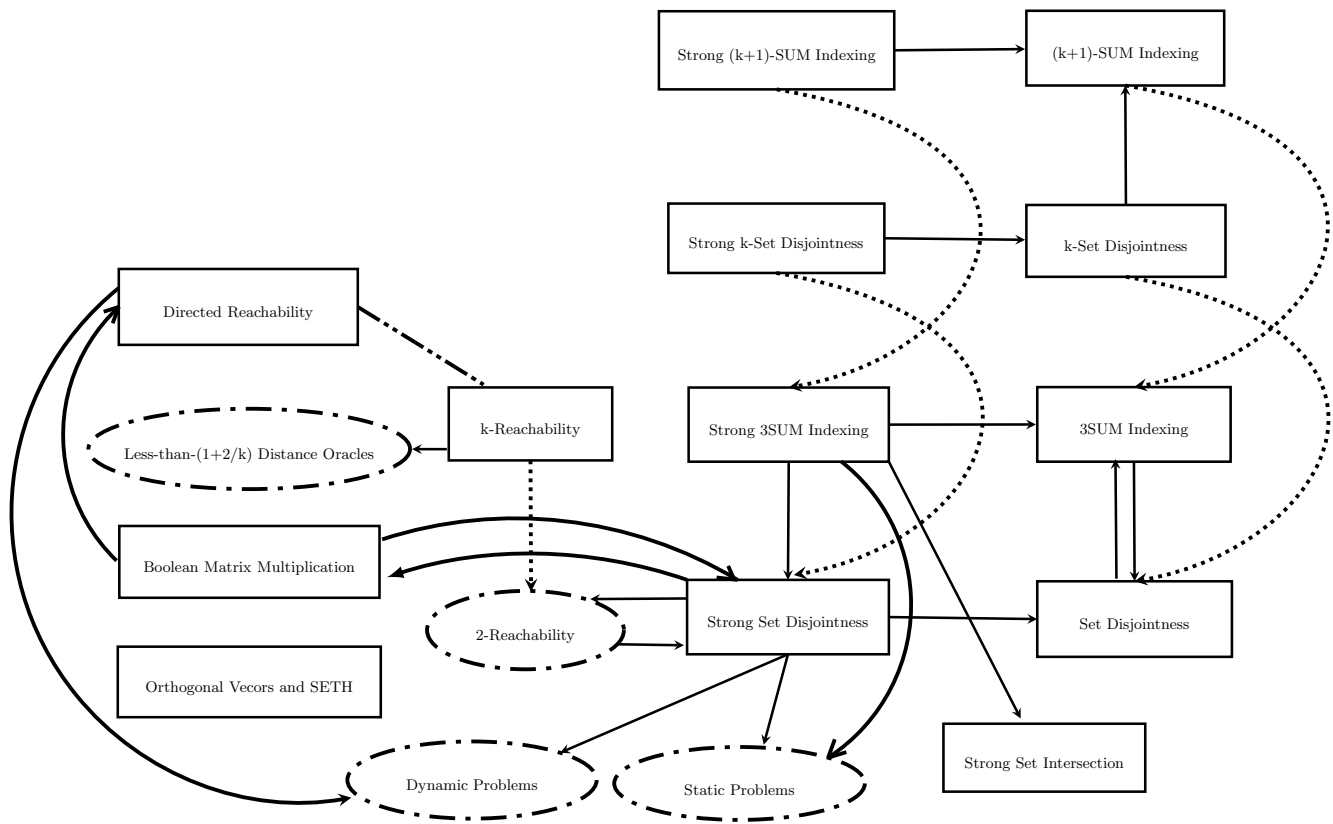}
\end{adjustbox}
\caption{\small\emph{Space conjectures and the connections between them as shown in this paper. Rectangles represent conjectures, while problems shown to be hard based on these conjectures are represented by ovals. Full arrow represents a reduction between two problems which also means an implication in the case of conjectures. Dotted arrow means a generalization which is also an implication, while dashed line means a generalization with no (known) reduction.}} \label{fig:M1}
\end{figure} 

\end{document}